\documentclass[12]{article}
\usepackage{color,graphics}
\usepackage{newinutile}

 \usepackage{epsfig}
 \usepackage{latexsym,amsmath,amsfonts,amsthm,amssymb,bbm}
 \usepackage{euscript,amsbsy}
 \usepackage{graphicx}
 \usepackage{caption}

 \newcommand{\bb}[1]{{\mathbb{#1}}}
 \newcommand{\bea}{\begin{eqnarray}}
\newcommand{\eea}{\end{eqnarray}}
\newcommand{\be}{\begin{equation}}
\newcommand{\ee}{\end{equation}}
\newcommand{\bangle}{\atopwithdelims \langle \rangle}

\newtheorem{theo}{Theorem}[section]

\newtheorem{lem}[theo]{Lemma}

\begin{document}


\title{Stationary currents in particle systems with constrained hopping rates}

\author{Emilio N. M. Cirillo}
\affiliation{Dipartimento di Scienze di Base e Applicate per
             l'Ingegneria\\ Sapienza Universit\`a di Roma,
             via A.\ Scarpa 16, 00161, Roma, Italy.}
\email{emilio.cirillo@uniroma1.it}
\thanks{ENMC expresses his thanks to ICMS (TU/e, The Netherlads) for
kind hospitality and financial support.}

\author{Matteo Colangeli}
\affiliation{Gran Sasso Science Institute, Viale F. Crispi 7, 
00167 L' Aquila, Italy.}
\email{matteo.colangeli@gssi.infn.it}

\author{Adrian Muntean}
\affiliation{Department of Mathematics and Computer Science,
             Karlstad University, Sweden.}
\email{adrian.muntean@kau.se}
	
\abstract{We study the effect on the stationary currents  of constraints
affecting the hopping rates in stochastic particle systems.
In the framework of Zero Range Processes with drift within a finite volume, we discuss how the current is
reduced by the presence of the 
constraint and deduce exact formulae, fully explicit in some
cases.
The model discussed here has been introduced in Ref.\ \cite{ccm2015}
and is relevant for the description of pedestrian motion 
in elongated dark corridors, where the constraint on the hopping rates can be related to
limitations on the interaction distance among pedestrians.
}


\maketitle

\section{Introduction}
\label{s:int}
This paper reports  on exact results for the calculation of stationary currents
for a class of one--dimensional zero--range
processes with threshold modeling the
dynamics of pedestrians walking in an elongated corridor with no visibility.
Modeling the corridor as a one--dimensional array of discrete sites, we assume
that more pedestrians (particles) can occupy the same site (forming, possibly,
social structures) and no interaction between these particles takes place.

The evolution of the pedestrians is determined by the level of occupancy of the
sites. The main specific feature is
the presence of the {\em activation threshold} which keeps the ``escape rate''
minimal until a certain occupation number on the site, corresponding to the threshold, is reached.
The threshold can be related to the limited interaction distance
among pedestrians Ref.\ \cite{ccm2015}:
only if a site is sufficiently populated pedestrians
can efficiently exchange information and move coherently to a neighboring
spot. The approach can
be further extended to consider the presence of multiple  thresholds (e.g. communication
saturation thresholds  cf.\ Ref.\ \cite{MMS}, de--centralized
task--allocation thresholds cf.\ Ref.\ \cite{ANTS}, and so on).
However, in that case of multiple thresholds exact
calculations are out of reach.
Our attempt is particularly relevant for the construction of exact microscopic and
macroscopic fundamental diagrams (explicit relationships between the
pedestrians speed and local density, see Ref. \cite{FD})
for pedestrians motion in one--dimensional models.

The distinguishing feature of the model is the presence of the 
activation threshold whose meaning for pedestrian 
motion has been discussed above
(see, also, Ref.\ \cite{CM2013} for the discussion of threshold 
effects in pedestrian dynamics in the framework of a two--dimensional 
model).
Nevertheless, different interpretations are possible: for instance, 
in the framework of 
{\em Porous Media}, the bulk porosity estimates how many particles can be 
accommodated in a cell and this connects to the activation threshold. 
Activation thresholds are also meaningful in pure 
{\em mechanical applications}: imagine that a device is equipped with 
valve--like door whose opening results from the balance between the 
pressure inside the cell and an external force exerted by a spring.
A minimal -- structural -- opening of the door, with the spring 
maintained at rest, corresponds to the presence of an activation threshold. 
A {\em psychologico--geometrical} interpretation is also possible: 
the activation threshold can be indeed regarded as a measure of the domain of communication between 
the individuals and the level this communication
is processed towards a 
decision on the motion. \\
Moreover, the mathematical framework developed in this work paves also the way for a deeper understanding, through the prism of stochastic dynamics, of the kinetic mechanisms giving rise to the hydrodynamic properties observed, e.g., in the study of the transport of a gas or a liquid through polymeric matrices, see e.g. \cite{George, Ohkubo}. Interestingly, note that despite the microscopic dynamics described in the sequel is not related to any energy function, the activation threshold present in our model connects, quite naturally, with the activation energy occurring in the Arrhenius expression for the rate of a chemical reaction, with the site occupancy (a random variable) playing the role of a temperature.\\
It is also worth mentioning that a suitable variant of the model discussed below was also introduced in the literature, see Ref. \cite{Evans96}, to investigate the thermodynamic properties of heterogeneous materials, in which, e.g., a single site may be equipped with a hopping rate whose dependence on the site occupancy differs from the rule assigned to the remaining sites. This situation was shown to give rise to interesting physical phenomena, cf. also Ref. \cite{Angel}.

Coming back to the original problem, 
the microscopic dynamics is 
modeled here via a Zero Range Process (ZRP), cf.
Ref.\ \cite{18},
in which the particles hop, with 
a certain intensity and an assigned probability, 
to the neighboring sites and in which the threshold affects the
intensity of the jumps from each lattice site.
In the framework of ZRP models, thresholds are not a novelty, see
e.g. Refs.\ \cite{CG15,GS}
where condensation and metastability effects
have been studied. In those papers the value of the threshold is
scaled with the size of the system and distinguishes between
``fast'' sites, namely those with a sufficiently small number of particles,
and ``slow'' sites, the remaining ones.

We exploit the threshold in a different fashion, see
Ref.\ \cite{ccm2015,MMS}: indeed, for our application,
the hopping rate must be increasing with the number of particles on
the spot and the threshold is used to activate the regime in which
the rate starts to increase linearly with the number of particles.
In Ref.\ \cite{ccm2015,MMS} the model has been studied in the
hydrodynamic limit, whereas in this Note we solve the model
for  finite values of the lattice size and the number of particles. 
In particular, here we compute the steady state current, which is, even in the
pedestrian motion interpretation, the main quantity of interest.

In the absence of threshold, the stationary current increases
proportionally to the number of particles, whereas it tends
to an asymptotic value when the threshold is equal to the number
of particles. In such a case no site exceeds the
threshold and the hopping rate stays always equal to its minimal
value. We compute the steady current for any intermediate
value of the threshold and, in particular, we prove that
thresholds proportional to the number of particles are sufficient
to induce the asymptotic limiting regime.

We remark that the focus of the paper is on the combined effect of a bias (i.e., a driving force, breaking the condition of detailed balance) and an activation threshold in presence of a \textit{finite number} of particles moving on a \textit{finite lattice}, endowed with periodic boundary conditions. Thus, we shed light on the finite size corrections to the value of the stationary current obtained in the hydrodynamic limit of the model (see Ref. \cite{16} for mathematical details): this program is pursued, here, by evaluating the canonical partition function, which can be explicitly read out in a few cases.

The paper is organized as follows. In Section~\ref{s:modello} we
introduce the model and define the stationary current.
In Section~\ref{s:canpart} we derive the expression of the partition
function of the model that is exploited in Section~\ref{s:currents}
to compute the current and to compare theoretical results to
numerical simulations.\\
Conclusions are drawn in Section~\ref{s:concl}.

\section{The model}
\label{s:modello}
\par\noindent
We consider a positive integer $L$ and define a ZRP on the finite torus 
$\Lambda:=\{1,\dots,L\}\subset\mathbb{Z}$, cf.\ Refs.\ \cite{16,15}.
We fix $N\in\bb{Z}_+$ and consider the finite \emph{state space}
$\Omega_{L,N}$:
\begin{equation}
\label{mod000}
\Omega_{L,N}
=
\Big\{
\eta\in\{0,\dots,N\}^\Lambda,\,
\sum_{x=1}^L\eta_x=N
\Big\}
\,\,.
\end{equation}
Given
$\eta=(\eta_1,\dots,\eta_L)\in\Omega_{L,N}$
the integer $\eta_x$ is called \emph{number of particle}
at site $x\in\Lambda$ in the \emph{state} or \emph{configuration}
$\eta$.
Pick the \textit{threshold} $T\in\{1,\dots,N\}$ and define the
\emph{intensity}
\begin{equation}
\label{soglia}
g_T(0)=0,\;\;\;
g_T(k)=1\;\; \textrm{ for } 1\leq k\le T,\;\;\;
g_T(k)=k-T+1\;\; \textrm{ for } k> T
\;\;\;.
\end{equation}

The ZRP considered in this paper
is the continuous time Markov process $\eta_t\in\Omega_{L,N}$, with $t\ge0$,
such that each
site $x\in\Lambda$ is updated with intensity $g_T(\eta_x(t))$
and, once such a site $x$ is chosen, a particle jumps to the
neighboring sites $x-1$ and $x+1$ with probabilities, respectively,
$1-p$ and $p$ (recall periodic boundary
conditions are imposed).
Note that the equilibrium condition of detailed balance holds only for $p=1/2$.
The above described jump process corresponds, hence, to an inhomogeneous Poisson process with hopping rates
\be
r_{x,x+1}(\eta)=g(\eta_x) p
,\;\;
r_{x,x-1}(\eta)=g(\eta_x) (1-p)
\;\;\textrm{ and }\;\;
r_{x,y}(\eta)=0
\;\;\textrm{ for } y\neq x-1,x+1
\;\;.
\label{hop}
\ee

Given the threshold $T$,
the intensity function is constantly equal to one up to $T$
and then it increases linearly with the number of particles occupying
the site. In other words, all sites with number of particles
smaller or equal to $T$ are treated equally by the dynamics,
whereas the updating of those sites with more than
$T$ particles is favored.
For this reason $T$ is called \emph{activation} threshold.

We note that in the limiting case $T=1$
the intensity function becomes
$g_1(k)=k$, for $k>0$, and thus
the well known
\emph{independent particle} model is recovered.
A different limiting situation is the one in which
the intensity function is constantly equal to $1$ for any $k\ge1$ and
equal to zero for $k=0$.
In this case a Zero Range
process whose configurations can be mapped to the
simple exclusion model states is found.
We shall refer to the latter case as to the
\emph{simple exclusion}--like model. Such a model is
found, in our set--up, when $T=N$ .
We stress that one of the interesting issues of our model is the
fact that it is able to tune between two very different dynamics:
the independent particle and simple exclusion--like behavior.

It can be proven (see Ref.\ \cite{18,16})
that the invariant measure of the ZRP process is a product measure of the form
\begin{equation}
\label{muinv}
\mu_{L,N,T}(\eta)
=
\frac{1}{Z_{L,N,T}}
\prod_{\substack{x=1,\dots,L:\\\eta_x\neq0}}
\frac{1}{g_T(1)\cdots g_T(\eta_x)}
\end{equation}
for any $\eta\in\Omega_{N,L}$, where the
\emph{partition function}
$Z_{L,N,T}$ is the normalization constant.

The main quantity of interest, in our study, is the \emph{stationary
current} representing the difference between the
average number of particles crossing a bond between two given
sites from the left to the right and that in the opposite direction.
More precisely, since periodic boundary conditions are imposed,
the current does not depend on the chosen bond and is defined as
\begin{equation}
\label{curr00}
J_{L,N,T}
:=\langle r_{x,x+1}-r_{x+1,x}\rangle_{L,N,T}
=(2p-1)\langle g_T\rangle_{L,N,T}
\end{equation}
where we introduced the notation
$\langle f(\eta)\rangle_{L,N,T}
 :=\sum_\eta \mu_{L,N,T}(\eta) f(\eta)$
for any function $f:\Omega_{L,N,T}\to\bb{R}$.

A general expression for the expectation, with respect to the
invariant measure \eqref{muinv}, of the intensity function can
be provided (see Ref.\ \cite{18}). More precisely it holds
\begin{equation}
\label {meang}
\langle g_T\rangle_{L,N,T}
=
\frac{Z_{L,N-1,T}}{Z_{L,N,T}}
\end{equation}
Equations
\eqref{curr00} and \eqref{meang}
yield the following expression
\begin{equation}
\label{curr}
J_{L,N,T}=
(2p-1)
\frac{Z_{L,N-1,T}}{Z_{L,N,T}}
\end{equation}
for steady state current.

\section{Canonical partition function}
\label{s:canpart}
The final goal of this paper is computing the steady
state current at finite volume for any value of the threshold.
In order to apply
equation \eqref{curr} we need an explicit expression
of the partition function.

In this Section we shall
prove an exact formula expressing the partition function
in terms of sums of factorials and
yielding
explicit expression of the partition function in the
limiting cases $T=1$ and $T=N$.

We first state a combinatorial lemma whose proof is based on
techniques borrowed from Ref.\ \cite{17}.
Given the positive integers $i,j,k$, we let
$\Phi(i,j,k)$ be the
\textit{number of ways in which}
$j$
\textit{indistinguishable balls can be distributed into} $i$
\textit{distinguishable urns with at most} $k$
\textit{balls into each urn}.
Note that for $j> ki$ we shall understand $\Phi(i,j,k)=0$.
For $i,j$ positive integers, we also let
\begin{equation}
\label{storto}
{i \bangle j}:=\binom{i+j-1}{j}
\end{equation}
which can be proven to be equal to the
number of ways in which
$j$
indistinguishable balls can be distributed into $i$
distinguishable urns, see Ref.\ \cite[section~3.2.12]{17}.

\begin{lem}
\label{lemmaPhi}
Let $i,j,k$ positive integers such that $j\le ki$, then
\be
\label{Phigen}
\Phi(i,j,k)=\sum_{s=0}^{\overline{s}} (-1)^s {i \bangle j-s(k+1)} \binom{i}{s}
\ee
where
$\overline{s}:=\min\{i,\lfloor j/(k+1)\rfloor\}$.
\end{lem}
We omit the simple proof of the Lemma \eqref{lemmaPhi}. 
It suffices to remark, here, that the proof relies on the theory of generating functions, as presented, e.g., in Ref. \cite[section 3.3.2]{17}.

The expression of $ \Phi(i,j,k)$ provided by the Lemma~\eqref{lemmaPhi}
attains a simpler form in a few cases.
For instance, when $j=k$, no constraint is imposed on the
allocation of balls among the urns, hence
one should find
\begin{equation}
\label{PhiBE}
\Phi(i,j,j)= {i \bangle j}
\end{equation}
This is indeed the case, since it holds $\bar{s}=0$.
Note that this is the result which is found when, in the Bose--Einstein
statistics, one counts the number of ways in which $j$ particles
can be distributed among $i$ states.
A second relevant case is the one in which at most one particle
can be allocated into each urn. The corresponding value of $\Phi$ may then be derived either from
Eq. \eqref{Phigen}, by using the fact that, since $k=1$ and $i\ge j$, it holds
$\overline{s}=\lfloor j/2\rfloor$, or from the combinatorial definition of $\Phi(i,j,1)$. In either case, one obtains
\begin{equation}
\label{PhiFD}
\Phi(i,j,1)=\binom{i}{j}
\end{equation}
Note that this is the result one encounters in the Fermi--Dirac
statistics, in which one counts the number of ways in which $j$ particles
can be distributed among $i$ states with the limitation, due to the
exclusion principle, of
at most one particle per state.

We can now state our main result about the canonical
partition function of the ZRP model.
Recall, see Eq. \eqref{muinv}, that
\be
\label{Z}
Z_{L,N,T}=\sum_{\eta:|\eta|=N} \prod_{\substack{x=1,\dots,L:\\\eta_x\neq0}}
\frac{1}{g_T(1)\cdots g_T(\eta_x)}
\;\;.
\ee

\begin{theo}
\label{main}
For $N,L$ positive integers
\be
\label{partfunFD}
Z_{L,N,1} = \frac{L^N}{N!}
\ee
Moreover, for any $T\ge2$ and $L\ge \lceil{N/(T-1)}\rceil$
\be
\label{partfun}
Z_{L,N,T}=\Phi(L,N,T-2)+\sum_{m=1}^{\overline{m}}\binom{L}{m}\sum_{n=0}^{\overline{n}}\Phi(L-m,N-[(T-1)m+n],T-2)\frac{m^n}{n!}
\ee
where
$\overline{m}= \lfloor N/(T-1)\rfloor$ and
$\overline{n}(m)=N-(T-1)m$.
\end{theo}

\begin{proof}
Consider, first, the case $T=1$. By \eqref{soglia} and
\eqref{Z} we get
\begin{displaymath}
\label{Z1}
Z_{L,N,1}
=
\sum_{\eta:|\eta|=N} \prod_{x=1,\dots,L} \frac{1}{\eta_x!}
\end{displaymath}
where we also used the convention $0!=1$. Thus,
equation \eqref{partfunFD} follows immediately by applying the
multinomial theorem, see, e.g., Ref.\ \cite[equation~(3.35)]{17}.

Consider, now, the case $T\ge2$.
Call $m\in[0,\overline{m}]$ the number of sites in which the number of
particles is larger than $T-1$ and
$n\in[0,\overline{n}]$ the number of particles that, for a given $m$,
exceeds the value $T$. Given a configuration $\eta$, let also
\be
n_x=
0  \;\;\textrm{ for } \eta_x\le T-1
\;\;\;\textrm{ and }\;\;\;
n_x=
\eta_x-(T-1)  \;\;\textrm{ for } \eta_x > T-1
\;\;.
\ee
Then, the partition function can be rewritten as
\begin{displaymath}
\label{Zproof}
Z_{L,N,T}=
\Phi(L,N,T-2)
+
\sum_{m=1}^{\overline{m}}\binom{L}{m}\sum_{n=0}^{\overline{n}}
\Phi(L-m,N-[(T-1)m+n],T-2)
\sum_{\substack{n_1+...+n_m=n\\n_x\geq0}} \frac{1}{n_1!~...~n_m!}
\end{displaymath}
The first term in Eq. \eqref{Zproof} takes into account the contribution to the sum
defining the partition function of those configurations in which
no site has a number of particles larger or equal to $T-1$.
The second term can be explained as follows:
the first binomial coefficient counts the number of ways one can
choose the $m$ sites such that $\eta_x\geq T-1$. 
Note that $\overline{m}$ denotes
the maximum value attained by $m$, for which it holds: $\overline{m}=\min(L,\lfloor N/(T-1)\rfloor)$. Yet, by requiring
$L\ge \lceil{N/(T-1)}\rceil$, one has $\overline{m}=\lfloor N/(T-1)\rfloor$, as indicated in the statement of the Theorem.
The coefficient $\Phi(L-m,N-[(T-1)m+n],T-2)$ counts
the number of ways
to allocate the remaining $N-[(T-1)m+n]$ particles on the $L-m$ sites for
which it holds $\eta_x\leq T-2$.
The last sum counts the number of ways in which the particles
exceeding $T$, namely, those on the top of the $T-1$ filled
columns, can be distributed on the $m$ sites.
Finally, recalling \eqref{soglia}, we have that
the last factor in the equation is a smart rewriting of the
last factor in \eqref{Z}.
Equation \eqref{partfun} in the theorem finally follows
by using the multinomial theorem (see, e.g., Ref.\ \cite[equation~(3.35)]{17}).
\end{proof}

We remark that, although \eqref{partfun} is not an explicit expression for
the partition function, it is nevertheless very useful.
Indeed, the sum over the configuration space present in the
definition of the partition function, Eq. \eqref{Z}, involves a number of terms
exponentially large in the number of particles $N$, whereas the
sum in \eqref{partfun} is only polynomial in $N$.
Moreover, the expression for partition function given in Eq. \eqref{Z} involves a constraint, namely $|\eta|=N$,
which has been removed in \eqref{partfun}.

It is also interesting to remark that in the simple exclusion--like
regime, namely, $T=N$, the partition function
$Z_{L,N,N}$ can be written explicitly as
\be
\label{partfunBE}
Z_{L,N,N}
=\Phi(L,N,N)
={L \bangle N}
=\binom{L+N-1}{N}
\;\;.
\ee
To prove this formula,
we compute, first, the term $\Phi(L,N,N-2)$ in \eqref{partfun}.
By using \eqref{Phigen} with $i=L$, $j=N$, and $k=N-2$,
noted that $\overline{s}=1$, one finds
\be
\label{Phic}
\Phi(L,N,N-2)
=\Phi(L,N,N)-{L \bangle 1} \binom{L}{1}
=\Phi(L,N,N)-L^2
\ee
Next, we evaluate the sums in Eq. \eqref{partfun}.
Since $\overline{m}=\overline{n}=1$,
one needs to calculate just the terms
$\Phi(L-1,0,N-2)$ and $\Phi(L-1,1,N-2)$.
Since, for both terms, $\overline{s}=0$ for both terms in
\eqref{Phigen}, we get
\be
\label{Phid}
\Phi(L-1,0,N-2)
= {L -1\bangle 0}
=1
\;\;\textrm{ and } \;\;
\Phi(L-1,1,N-2)
={L-1 \bangle 1}
=L-1
\ee
Equation~\eqref{partfunBE} finally follows from
\eqref{Phic}, \eqref{Phid}, and \eqref{partfun}.

Note that the result in Eq. \eqref{partfunBE} could also be directly
deduced by equation \eqref{Z}. Indeed, from the definition
\eqref{soglia} of the intensity function, it follows
that, for $T=N$, the sum in equation \eqref{Z} is indeed
a sum of $1$'s and, thus, yields
straightforwardly the total number of configurations $\Phi(L,N,N)$.

As discussed in the Introduction, the threshold limits the hopping rate on sites whose occupancy number is smaller than the
threshold itself, whereas, when the prescribed value of the threshold is reached, the hopping
rate starts increasing proportionally to the number of particles on the
site.
In this respect, the case $T=N$ is peculiar, because all the sites are updated with the same minimal rate
regardless their occupancy number.

It is also possible
to guess another remarkable result: namely, when the threshold,
although smaller than $N$, scales proportionally to $N$, then the stationary current is close, for large $N$, to the value obtained for $T=N$.
More precisely, take $\alpha<1$, $\alpha\in \mathbb{R}^{+}$ and sufficiently close to $1$, and compare the
canonical partition function of the systems with $N$ particles and thresholds equal, respectively, to $\alpha N$ and $N$.\\
We thus conjecture that for $N\to\infty$
\begin{equation}
\label{comp}
\frac{Z_{L,N,\alpha N}}{Z_{L,N,N}}=1+o(1) \nonumber
\end{equation}
where $o(1)$ denotes a function tending to zero in the limit $N\to\infty$.

We omit, here, the lengthy algebraic details, and we just mention
that this observation may be relevant in the study of the hydrodynamic limit of heterogeneous ZRP, in which the hopping rate from a given site
can be modified so as to scale with the size of the system. 

\section{Stationary currents and numerical simulations}
\label{s:currents}
In this Section we report and compare both
analytical and numerical results for the steady current in the ZRP with threshold
introduced in Sec. \ref{s:modello}.

Numerics have been performed via Monte Carlo techniques
by simulating the model as follows: call $\eta(t)$ the
configuration at time $t$,
a number $\tau$ is chosen at random with
exponential distribution of
parameter $\sum_{x=1}^Lg_T(\eta_x(t))$ and time is update to
$t+\tau$,
a site is chosen at random with probability
$g_T(\eta_x(t))/\sum_{x=1}^Lg_T(\eta_x(t))$ and a particle
is moved from such a site to its right with probability $p$ and
to its left with probability $1-p$.

The Monte Carlo simulation is let, first, evolve for a number of time steps $n_{0}\sim10^7$, and
the stationary current is thus defined as the ratio of the difference
between the total number of particles jumping from site $L$ to site $1$
and that of particles jumping from site $1$ to site $L$,
to the total time.
We remark that the initial number of time steps $n_{0}$ is chosen large enough to guarantee that  a constant value, with respect to
time, is reached by the current.

As for the analytical results on the current, note that
the theory developed in the Sections above
paves the way to the computation of the stationary current
for any finite value of the parameters of the model, $N$, $L$, and $T$.
We stress that we are considering a transport problem in which 
a net convective flux occurs in the case $p\neq 1/2$.
Equations \eqref{curr} and \eqref{partfun} can be used to reduce the computation of the
stationary current to an algebraic sum.
In particular, in the two limiting cases $T=1$ and $T=N$ analytic formulae can be derived.

Indeed,
from Eqs. \eqref{curr} and \eqref{partfunFD}, the steady state current for
$T=1$, i.e. in the independent particle case, reads
\be
\label{curr1}
J_{L,N,1}=(2p-1) \frac{N}{L}
\ee
On the other hand, Eqs. \eqref{curr} and \eqref{partfunBE}
imply that, for $T=N$, i.e. in the simple exclusion--like
regime, the current is given by
\be
\label{currN}
J_{L,N,N}=(2p-1) \frac{N}{L}\frac{1}{1+\frac{N}{L}-\frac{1}{L}}
\ee
We stress that the two results above are valid for any
finite volume $\Lambda$ and for any finite number of particles.
If the limit $N,L\to\infty$ with $N/L=\varrho$ is considered,
the well-known hydrodynamic limit is found for the current,
see Ref.\ \cite[equation (1.3)]{CR1997}.

\begin{figure}[htbp!]
\begin{center}
\includegraphics[width=0.4\textwidth]{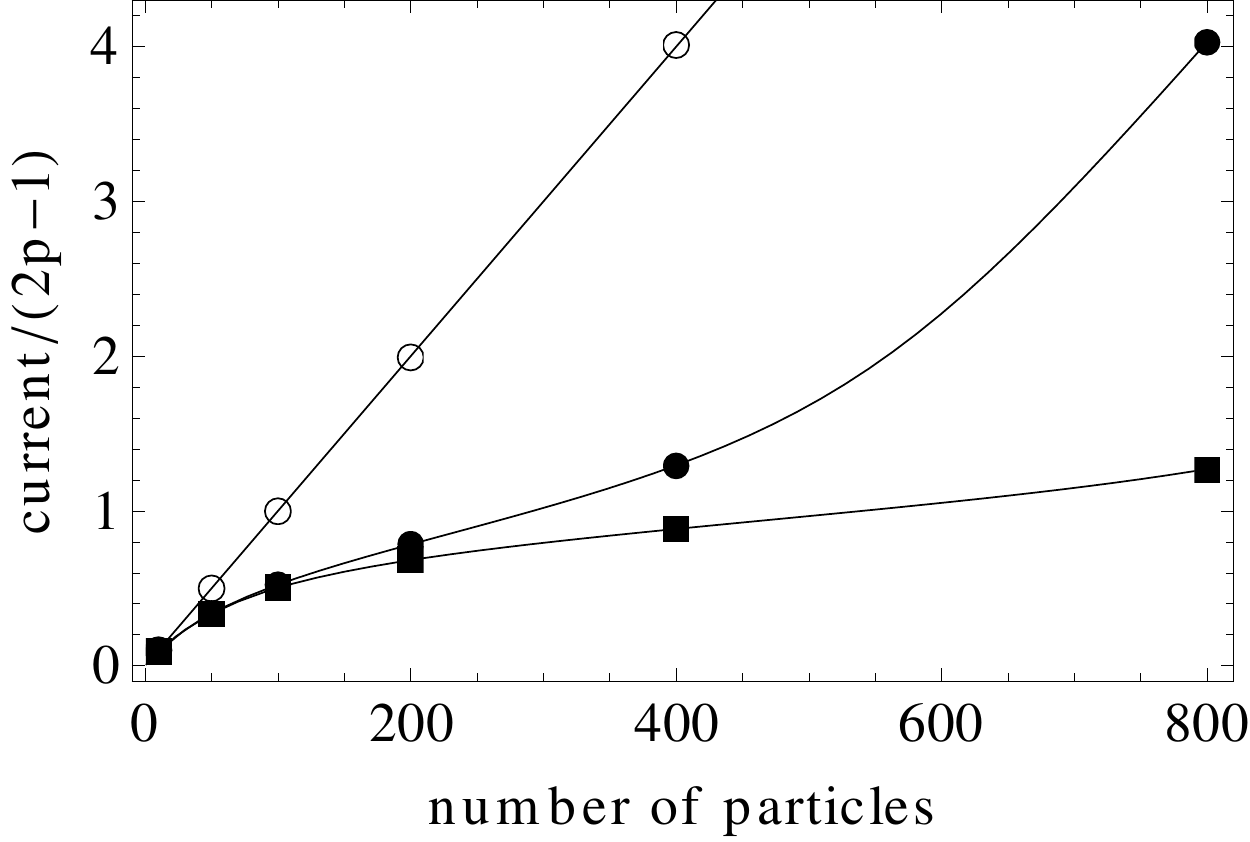}
\hskip 1 cm
\includegraphics[width=0.4\textwidth]{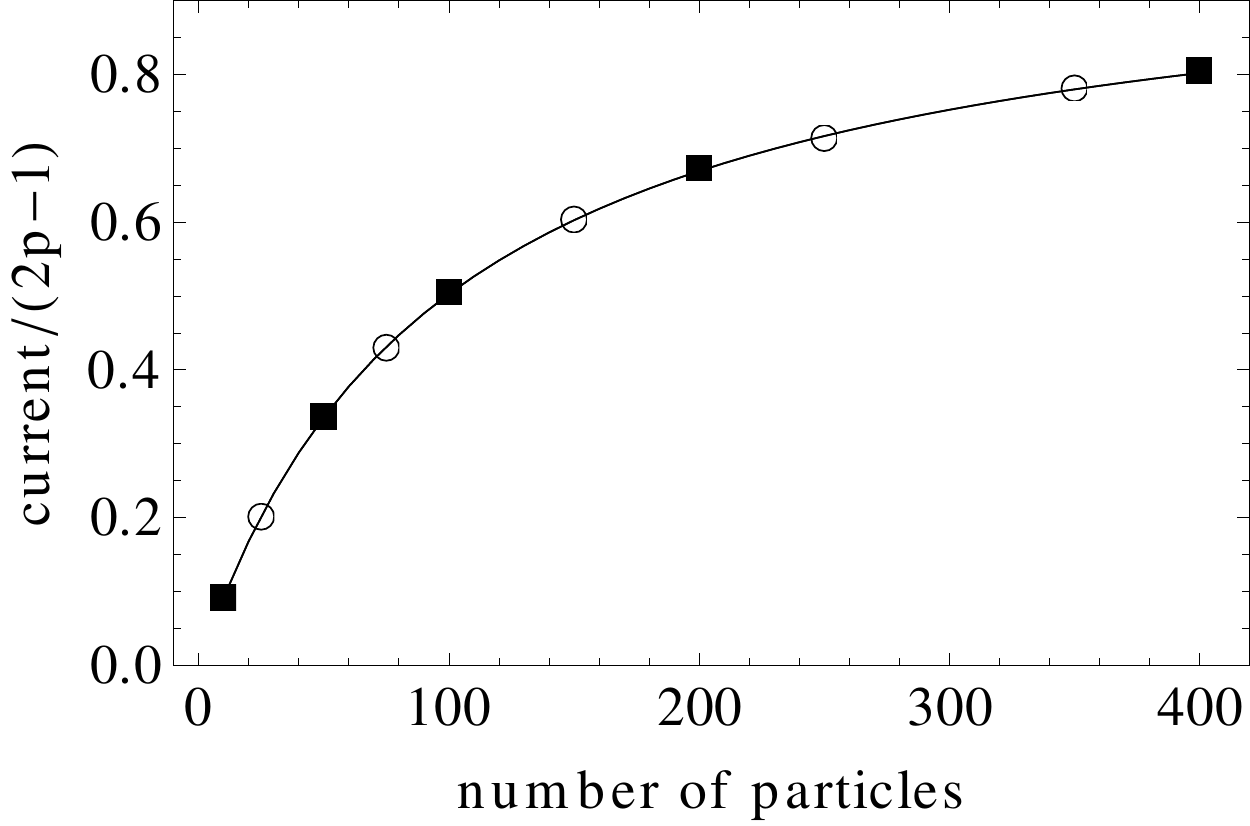}
\caption{Current versus number of particles for $L=100$.
\textit{Left panel}:
Open circles, solid circles, and solid squares denote,
respectively, the simulated stationary current for the
threshold $T=1,5,10$.
The associated solid lines represent the analytic solutions:
Eqs. \eqref{curr} and \eqref{partfun} have been used in the cases
$T=5$ and $T=10$, whereas the explicit formula in Eq. \eqref{curr1} has been
used for $T=1$.
\textit{Right panel}:
Open circles and solid squares denote,
respectively, the simulated stationary current for the
threshold $T=N$ and $T=N/2$.
The associated solid lines represent the analytic solutions
computed by using
Eqs. \eqref{curr} and \eqref{partfun} for $T=N/2$ and
\eqref{currN} for $T=N$.
The two solid lines are not distinguishable in the picture.
}
\label{fig:t1510}
\end{center}
\end{figure}

Coherently with its physical interpretation, the effect of the
activation threshold is that of slowing down the current.
As $T$ is increased the steady state current decreases.
In particular it is worth mentioning that at $T=1$ the
current is a linear function of the number of
particles $N$, whereas at $T=N$ the current saturates to a limiting
value when $N$ is increased.
For intermediate thresholds, namely, $1<T<N$, the current
increases slowly with $N$ and only after a certain value it
starts growing linearly.
This effect is clearly illustrated in Fig.~\ref{fig:t1510} (left panel)
where the current is plotted versus the total number of
particle for the values $T=1,5,10$ of the activation threshold and
$L=100$.

Data in Fig.~\ref{fig:t1510} (right panel)
refer to the cases $T=N$ and $T=N/2$, with $L=100$.
The saturation effect on the current due to the presence
of the threshold $T=N$ (simple exclusion--like regime) is clearly illustrated.
In other words, when the hopping rate is constantly equal to one and does
not depend on the number of particles on the site, the current
tends to saturate to a constant value for $N$ large.
It is worth remarking that the same effect
is also observed when the threshold is equal to $N/2$, suggesting
that for an activation threshold increasing proportionally to the
number of particles, the current is reduced in the same fashion.
This property is indeed an immediate consequence of Eqs.
\eqref{curr} and \eqref{comp}.
Deviation from the $T=N$ behavior in the case $T=N/2$ can in fact be
observed for small values of $L$ and $N$.

\section{Conclusions}
\label{s:concl}
\par\noindent
We considered the problem of computing the steady state current in a Zero Range Process subjected to a drift
as well as to an ``activation'' threshold affecting the hopping rates of the particles to the neighboring sites.
By exploiting combinatorial arguments, 
we derived an exact formula for the partition function,
which is amenable to an analytical treatment for $T=1$ and $T=N$.
We also discussed the asymptotic behavior of the partition function when the threshold
scales proportionally to the number of particles: the latter case is of particular relevance in the discussion
of the hydrodynamic limit of the model.
We then obtained explicit formulae for the particle current, also supported by Monte Carlo simulations,
revealing that the main effect of the activation threshold on the steady state dynamics is to decrease the current,
thus tuning between two limiting regimes, the independent particle model and
the simple exclusion--like process.
We also remarked that this last behavior is shown by the model 
even for $T<N$, provided the threshold increases proportionally 
to the number of particles.

\begin{acknowledgements}
ENMC expresses his thanks to ICMS (TU/e, The Netherlands) for
kind hospitality and financial support.
\end{acknowledgements}

\end{document}